\documentclass[article]{IEEEtran}


\usepackage{graphicx,colordvi,psfrag}
\usepackage{amsmath,amssymb}
\usepackage{epstopdf}
\usepackage[caption=false]{subfig}
\usepackage{epsfig,cite}
\usepackage{calc,pstricks, pgf, xcolor}
\usepackage{nicefrac}
\usepackage{enumerate}

\newtheorem{theorem}{Theorem}
\newtheorem{corollary}{Corollary}
\newtheorem{proposition}{Proposition}
\newtheorem{remark}{Remark}

\newenvironment{proof}[1][Proof]{\noindent\textbf{#1.} }{\ \rule{0.5em}{0.5em}}

\newcommand{\sd}{\footnotesize{\Sigma\Delta}}
\newcommand{\dpcm}{\text{DPCM}}

\newcommand{\bx}{\mathbf{x}}

\newcommand{\ZZ}{\mathbb{Z}}

\newcommand{\Unif}{\mathop{\mathrm{Uniform}}}

\newcommand{\argmin}{\operatornamewithlimits{argmin}}
\newcommand{\m}{\mathcal}

\newrgbcolor{BoxRed}{1 .5 .5}
\newrgbcolor{BoxBlue}{.9 .9 1}

\begin{document}

\title{Performance Analysis and Optimal Filter Design for Sigma-Delta Modulation via Duality with DPCM}

\author{Or~Ordentlich and
        Uri~Erez,~\IEEEmembership{Member,~IEEE}
\thanks{The work of O. Ordentlich was supported by the Adams Fellowship Program of the Israel Academy of Sciences and Humanities, a fellowship from The Yitzhak and Chaya Weinstein Research Institute for Signal Processing at Tel Aviv University and the Feder Family Award. The work of U. Erez was supported by by the ISF under Grant 1557/13.}
\thanks{O. Ordentlich and U. Erez are with Tel Aviv University, Tel Aviv, Israel (email: ordent,uri@eng.tau.ac.il).
}}

\maketitle

\begin{abstract}
Sampling above the Nyquist rate is at the heart of sigma-delta modulation, where the increase in sampling rate is translated to a reduction in the overall (mean-squared-error) reconstruction distortion. This is attained by using a feedback filter at the encoder, in conjunction with a low-pass filter at the decoder. The goal of this work is to characterize the optimal trade-off between the per-sample quantization rate and the resulting mean-squared-error distortion, under various restrictions on the feedback filter. To this end, we establish a duality relation between the performance of sigma-delta modulation, and that of differential pulse-code modulation when applied to (discrete-time) band-limited inputs. As the optimal trade-off for the latter scheme is fully understood, the full characterization for sigma-delta modulation, as well as the optimal feedback filters, immediately follow.
\end{abstract}

\section{Introduction}
\label{sec:Into}

Analog-to-digital (A/D) and digital-to-analog (D/A) converters are essential in modern electronics. In many cases, it is the quality of these converters that constitutes the main bottleneck in the system, and consequently, dictates its entire performance. On the other hand, as digital circuits are now considered relatively cheap to implement, the interface between the analog and digital domains is often one of the most expensive components in the system. Developing A/D and D/A components that are on the one hand relatively simple, and on the other hand introduce little distortion, is therefore of interest.

Often, the same A/D (or D/A) component is applied to a variety of signals with distinct characterizations. For this reason, it is desirable to design the data converter to be robust to the characteristics of the input signal. One assumption that cannot be avoided is the bandwidth of the signal to be converted, which dictates the minimal sampling rate, according to Nyquist's theorem. Beyond bandwidth, however, one would like to assume as little as possible about the input signal. A reasonable model for the input signal is therefore a \emph{stochastic} one, where the input signal is assumed to be a stationary Gaussian process with a given variance and an arbitrary \emph{unknown} power spectral density (PSD) within the assumed bandwidth, and zero otherwise. In this paper, we adopt this \emph{compound} model which is rich enough to include a wide variety of processes. The robustness requirement from the A/D (or D/A) converter translates to requiring that it induces a small average distortion simultaneously for all processes within our compound model.

Sigma-delta modulation is a widely used technique for A/D as well as D/A conversion. The main advantage offered by this type of modulation is the ability to trade-off the sampling rate and the number of bits per sample required for achieving a target mean-squared error (MSE) distortion. The input to the sigma-delta modulator is a signal sampled at $L$ times the Nyquist rate ($L>1$). This over-sampled signal is then quantized using an $R$-bit quantizer. In much of the literature about sigma-delta modulation, no stochastic model is assumed for the input signal. However, when such a model is assumed, the benefit of over-sampling can be easily understood from basic rate-distortion theoretic principles: the (per-sample) rate required to achieve distortion $D$ for the over-sampled signal is $L$ times smaller than the rate required to achieve the same distortion for the signal obtained by sampling at the Nyquist rate. Thus, in principle, increasing the sampling rate should allow one to use quantizers with lower resolution, which is desirable in many applications.

However, the rate-distortion theoretical property that guarantees a constant product of the number of bits per sample needed to achieve distortion $D$, and the over-sampling ratio $L$, is only valid when a very long block of samples is vector-quantized. In A/D and D/A conversion, vector-quantization in high dimensions is a prohibitively complex operation, and quantization is invariably done via scalar uniform quantizers. Scalar quantizers alone cannot translate the increase of sampling rate to a significant reduction in the necessary resolution, but fortunately this problem can be circumvented with the aid of appropriate signal processing.

In sigma-delta based converters, the quantization noise is shaped using a causal shaping filter embedded within a feedback loop, see Figure~\ref{fig:SigDelTest}. The filter coefficients are chosen in a manner that ensures that most of the energy of the shaped quantization noise lies outside the frequency band occupied by the over-sampled signal. At the decoder, the quantized signal is low-pass filtered, cancelling out the high-frequencies of the quantization noise process without effecting the signal, such that the decoder's output is composed of the original signal corrupted by a low-pass noise process. %In order to follow the ideal rate-distortion function, the energy of the shaped quantization noise within the signal band should be smaller than the original energy of the (non-filtered) quantization noise by a factor of $2^{-2L}$.

Another technique for compressing sources with memory, which explicitly models the source as a stochastic process, is differential pulse-code modulation (DPCM). In DPCM, a prediction filter is applied to the quantized signal. The output of this filter is then subtracted from the source and the result is fed to the quantizer, see Figure~\ref{fig:DPCMtest}. At the decoder, the quantized signal is simply passed through the inverse of the prediction filter. The well-known ``DPCM error identity''~\cite{jayantnoll} states that the output of the decoder is equal to the source plus the quantization error, just like in simple non-predictive quantization. The benefit of using DPCM, however, is that the signal fed to the quantizer is the error in predicting the source from its \emph{quantized} past, rather than the source itself. If the coefficients of the prediction filter are chosen appropriately, the variance of this error should be smaller than the variance of the original source, which translates to a reduction in the number of bits required from the quantizer for achieving a certain distortion.

The performance of DPCM under the assumption of high-resolution quantization is well understood since as early as the mid 60's~\cite{McDonald66,jayantnoll,gn98}. Under this assumption, the prediction filter should be chosen as the optimal linear minimum mean-squared-error (MMSE) prediction filter of the source process from its past~\cite{jayantnoll}, and the effect of the filtered quantization noise can be neglected in the prediction process. While in most cases where DPCM is traditionally used, the high resolution assumption is well justified, it totally breaks down for the class of band-limited processes, which includes the input signals to sigma-delta modulators. Indeed, the prediction error of such a process from its infinite past has zero-variance, rendering the DPCM high-resolution rate-distortion formulas completely useless.

\subsection{Connection to Previous Work}

The connection between DPCM and sigma-delta modulation, as two instances of predictive coding, was known from the outset. Indeed, both paradigms emerged from two Bell-Labs patents authored by CC Cutler~\cite{cutler1952,cutler1954} in 1952 and 1954.

In fact, by adding appropriate pre- and post-filters to the sigma-delta modulator, as depicted in Figure~\ref{fig:PrePotsSigDelTest}, the input to the quantizer, as well as the final reconstruction of the signal, become identical to those in the DPCM architecture~\cite[Section II]{th78},~\cite[Chapter 3.2.4]{derpichPhD}. For this reason, it has become folklore that the two architectures are equivalent. When a sigma-delta modulator is used for compression of digital discrete-time signals, the pre-filtering can be performed digitally and the additional complexity of the architecture depicted in Figure~\ref{fig:PrePotsSigDelTest}, w.r.t. that in depicted in Figure~\ref{fig:SigDelTest}, may be acceptable. This is however \emph{not} the case for data converters, as the input to the latter is analog and pre-filtering must be done in continuous-time, which is more challenging. The motivation for this work is understanding the performance limits of A/D and D/A conversion based on the sigma-delta architecture, and therefore pre-filtering is precluded. Thus, the architecture is confined to that depicted in Figure~\ref{fig:SigDelTest}.

Another important aspect of our interest in sigma-delta modulators as a mean of data-conversion rather then data-compression, is that it dictates that the assumptions one can make on the statistics of the input signal must be minimal. Consequently, we consider a \emph{compound} class of sources that consists of all stationary Gaussian processes with variance $\sigma^2_X$ whose PSD is limited to some predefined frequency band. In addition, since data converters often operate at very high rates, it makes sense to impose various constraints on the sigma-delta feedback filter $C(Z)$, such as confining it to be a finite impulse response (FIR) filter with a limited number of taps. For a given desired MSE distortion level, our goal is to find the constrained sigma-delta feedback filter $C(Z)$ that minimizes the quantization rate w.r.t. all sources in the compound class, and to characterize the attained rate. This goal is different than the one pursued in~\cite{dsqg08}, where the optimal \emph{unconstrained} filters w.r.t. a known PSD were found.

The problem of finding the optimal $N$-tap FIR sigma-delta feedback filter $C(Z)$ for a compound family of sources similar to ours, was considered in~\cite{th78}. The optimal filter was claimed in~\cite{th78} to be the $N$th order MMSE prediction filter $C(Z)=(1-Z^{-1})^N$ of a bandpass stationary process from its past, and for a fixed target MSE distortion the required quantization rate was found to decrease linearly with $N$. Such as statement is obviously inaccurate, as it violates Shannon's rate-distortion theorem. The major drawback of~\cite{th78} is that it (implicitly) makes the high-resolution assumption that the variance of the quantizer's input is solely dictated by the target signal $\{X_n\}$, whereas the contribution of the quantization noise to this variance can be neglected. As discussed above, for over-sampled processes this assumption may not be valid even when the quantizer's resolution is very high. In particular, using the filter $C(Z)=(1-Z^{-1})^N$ from~\cite{th78}, the energy of the quantization noise within the frequency band occupied by the signal indeed decreases exponentially with $N$. However, the noise's energy outside this band increases rapidly with $N$, and for any quantization resolution it will become much greater than $\sigma^2_X$ for $N$ large enough, making the high-resolution assumption inapplicable. In this case, the dynamic range of the quantizer will be exceeded and overload errors would frequently occur.

It therefore follows that in the analysis of sigma-delta modulators one should not make high-resolution assumptions, but rather must take into account the effect of the filtered quantization noise on the variance of the quantizer's input. Fortunately, in the analysis of DPCM modulators the high-resolution assumption has been overcome in~\cite{zke08}. It was shown that for any distortion level and any stationary Gaussian source, the DPCM architecture induces a rate-distortion optimal test channel, provided that the prediction filter is chosen as the optimal filter for predicting the source from its \emph{quantized past}, and in addition water-filling pre- and post-filters are applied. The analysis of~\cite{zke08}, which takes into account the effect of the quantization noise, can therefore be used to obtain the optimal feedback filter and its corresponding performance for a DPCM system applied to an over-sampled stationary Gaussian source. In this paper, we leverage the results from~\cite{zke08} to the analysis of sigma-delta modulators, by establishing an appropriate duality between the two architectures.

\subsection{Contributions}
\label{subsec:contributions}

Let $\m{S}$ be the compound class of all discrete-time stationary Gaussian sources with variance $\sigma^2_X$ and PSD that is zero for all $\omega\notin [-\pi/L,\pi/L]$, $L\geq 1$. Note that this class corresponds to uniformly sampling a compound class of continuous-time stationary Gaussian processes with variance $\sigma^2_X$ and PSD that is zero for all $|f|>f_{\text{max}}$, at a sampling rate of $2Lf_{\text{max}}$ samples/per second. Let $\{X^{\dpcm}_n\}$ be a discrete-time stationary Gaussian process with PSD
\begin{align}
S^{\dpcm}_X(\omega)=\begin{cases}
L\sigma_X^2 & \text{for } |\omega|\leq \pi/L \\
0 & \text{for } \pi/L<|\omega|<\pi
\end{cases},\label{Xpsd}
\end{align}
and note that $\{X^{\dpcm}_n\}\in\m{S}$.

Our main result, derived in Section~\ref{sec:mainresult}, is that for any process $\{X_n^{\sd}\}$ from the compound class $\m{S}$, the test channel induced by the sigma-delta modulator (Figure~\ref{fig:SigDelTest}) achieves exactly the same rate-distortion function as that of the DPCM test channel (Figure~\ref{fig:DPCMtest}) with input $\{X^{\dpcm}_n\}$. More specifically, for such processes, for any choice of $\sigma^2_{\text{DPCM}}$ and prediction filter $C(Z)$ in the test channel of Figure~\ref{fig:DPCMtest}, the same choice of $C(Z)$ together with the choice
\begin{align}
\sigma^2_{\Sigma\Delta}=\frac{\sigma^2_{\text{DPCM}}}{L \cdot \frac{1}{2\pi}\int_{-\pi/L}^{\pi/L}|1-C(\omega)|^2 d\omega}\label{sigdelnoise}
\end{align}
in Figure~\ref{fig:SigDelTest}, yields the same compression rate and the same distortion.

While this result is simple to derive, it has a very pleasing consequence: the problem of optimizing the filter $C(Z)$ in sigma-delta modulation w.r.t. any signal in $\m{S}$, under any set of constraints, can be cast as an equivalent problem of optimizing the DPCM prediction filter w.r.t. input $\{X^{\dpcm}_n\}$ under the same set of constraints. Furthermore, in Section~\ref{subsec:fwmse}, we formalize a similar duality between DPCM and sigma-delta modulation for a frequency-weighted-mean-squared-error distortion measure. In this case $S^{\dpcm}_X(\omega)$ is replaced with a PSD that depends on the distortion's weight function.

In principle, recasting the sigma-delta optimization problem as an MMSE prediction problem may be derived directly from the formulas characterizing its performance, as given in Proposition~\ref{prop:SigDelRD}. Nevertheless, establishing the equivalence between sigma-delta modulation and DPCM, in the specific form described above, is insightful as it allows to borrow known results from the literature about the latter.

Having recast the filter optimization problem for sigma-delta as that of optimal linear prediction, we can readily obtain the solution under constraints for which an explicit solution was lacking in the literature, or was cumbersome to derive.

One may question the relevance of the test channel of Figure~\ref{fig:SigDelTest} and its information-theoretic analysis to the practical, resource limited, problem of A/D and D/A conversion via sigma-delta modulators. %In Section~\ref{sec:vecquantization} we follow~\cite{zke08,oz09} and show that if a high-dimensional vector quantizer is applied in sigma-delta modulation, in conjunction with an appropriately chosen interleaver, the same performance predicted by the information-theoretic analysis in Section~\ref{sec:mainresult} is obtained. The vector quantizer used in Section~\ref{sec:vecquantization} is only required to be a good quantizer for an i.i.d. Gaussian source. However, its high dimension makes it inapplicable for A/D and D/A conversion.
To that end, in Section~\ref{sec:scalarquant} we replace the AWGN channel from Figure~\ref{fig:SigDelTest} with a simple scalar uniform (dithered) quantizer of finite support, which is suitable for implementation within A/D and D/A converters. As long as overload does not occur, the effect of applying the scalar quantizer is equivalent to that of an additive noise channel. We show that the rate-distortion trade-off derived for sigma-delta modulation in Section~\ref{sec:mainresult} remains valid with high probability, with a constant additive excess-rate penalty for using scalar quantization. The purpose of this excess-rate is to ensure that an overload event, which jeopardizes the stability of the system, occurs with low probability. The stochastic model we assume for the input process allows us to tackle the issue of stability in a systematic and rigourous manner, and the trade-off between the excess-rate and the overload probability is analytically determined.

Clearly, a sigma-delta modulator can only perform well if overload errors are rather rare. Our stability analysis in Section~\ref{sec:scalarquant} is based on avoiding overload events w.h.p., and does not aim to consider the effect of such events on the distortion once they occur. In general, the overload probability of the scheme described in Section~\ref{sec:scalarquant} decreases double exponentially with the excess-rate of the quantizer w.r.t. the mutual information. Thus, taking an excess rate of $1-2$ bits will usually yield a sufficiently low overload probability. However, sigma-delta quantizers are often employed with a one-bit quantizer. In this case, the overload error probability cannot be very low. Consequently, the designer would need to guarantee that the effect of overload errors is local in time, and does not drive the system out of stability. There are various restrictions one can place on $C(Z)$ in pursuit of the latter goal. The issue of maintaining stability when overload errors are unavoidable is outside the scope of this paper. Nevertheless, we stress that our main result is of great relevance to this setting, as it shows that the filter $C(Z)$ should be chosen as the optimal MMSE prediction filter of $\{X^{\dpcm}_n\}$ from its noisy past under the stability ensuring restrictions.

%In Section~\ref{sec:numerical} we demonstrate the introduced ideas via several examples.

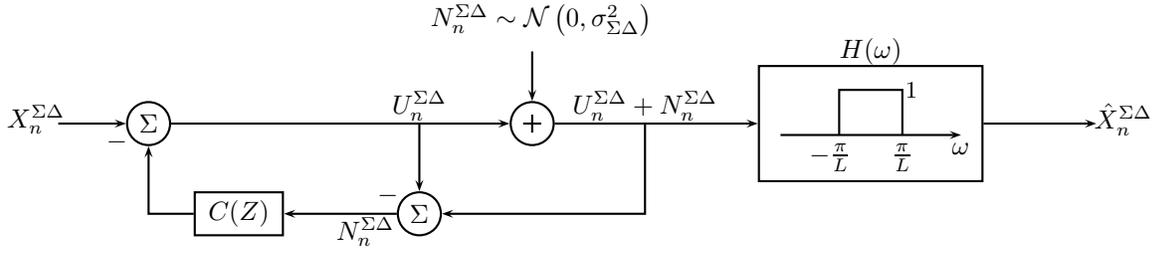
\begin{figure*}[]
\begin{center}
\psset{unit=0.6mm}
\begin{pspicture}(0,0)(250,55)
\rput(0,30){
\rput(0,1){$X^{\sd}_n$}\psline{->}(5,0)(20,0)\pscircle(25,0){5}\rput(25,0){$\Sigma$}\rput(18,-4){$-$}\psline{->}(30,0)(105,0)\rput(85,4){$U^{\sd}_n$}
\pscircle(110,0){5}\psline(108,0)(112,0)\psline(110,-2)(110,2)\psline{->}(110,16)(110,5)
\rput(112,23){$N^{\sd}_n\sim\mathcal{N}\left(0,\sigma^2_{\Sigma\Delta}\right)$}%{$N_n\sim\mathcal{N}\left(0,\frac{D}{\frac{1}{2\pi}\int_{-\pi/L}^{\pi/L}|1-C(\omega)|^2 d\omega}\right)$}
\psline{->}(115,0)(160,0)\rput(135,4){$U^{\sd}_n+N^{\sd}_n$}\psframe(160,-13)(210,13)\rput(185,16){$H(\omega)$}
\psline{->}(135,0)(135,-20)(90,-20)\pscircle(85,-20){5}\rput(85,-20){$\Sigma$}\rput(78,-16){$-$}\psline{->}(85,0)(85,-15)
\psline{->}(80,-20)(55,-20)\rput(73,-24){$N^{\sd}_n$}
\psframe(35,-25)(55,-15)\rput(45,-20){$C(Z)$}\psline{->}(35,-20)(25,-20)(25,-5)
\rput(160,2.5){
\psline{->}(5,-5)(45,-5)\rput(45,-8){$\omega$}
\psline(18,-5)(18,5)(32,5)(32,-5)\rput(34,5){\small $1$}
\rput(16,-10){$-\frac{\pi}{L}$}\rput(32,-10){$\frac{\pi}{L}$}
}
\psline{->}(210,0)(235,0)\rput(241,1){$\hat{X}^{\sd}_n$}
}
\end{pspicture}
\end{center}
\caption{The test channel corresponding to the sigma-delta modulation architecture, with the sigma-delta quantizer replaced by an AWGN channel. The input is assumed to be over-sampled at $L$ times the Nyquist rate.}
\label{fig:SigDelTest}
\end{figure*} 
\begin{figure*}[]
\begin{center}
\psset{unit=0.6mm}
\begin{pspicture}(0,0)(250,65)
\rput(0,40){
\rput(0,1){$X^{\dpcm}_n$}\psline{->}(5,0)(20,0)\pscircle(25,0){5}\rput(25,0){$\Sigma$}\rput(18,-4){$-$}\psline{->}(30,0)(55,0)\rput(45,4){$U^{\dpcm}_n$}
\pscircle(60,0){5}\psline(58,0)(62,0)\psline(60,-2)(60,2)\psline{->}(60,16)(60,5)\rput(62,20){$N^{\dpcm}_n\sim\mathcal{N}\left(0,\sigma^2_{\text{DPCM}}\right)$}%{$N_n\sim\mathcal{N}\left(0,L\cdot D \right)$}
\psline{->}(65,0)(100,0)\rput(85,4){$U^{\dpcm}_n+N^{\dpcm}_n$}\pscircle(105,0){5}\rput(105,0){$\Sigma$}\rput(98.5,-4){$+$}
\psline{->}(110,0)(160,0)\psframe(160,-13)(210,13)\rput(185,16){$H(\omega)$}\rput(135,4){$V^{\dpcm}_n$}
\rput(160,2.5){
\psline{->}(5,-5)(45,-5)\rput(45,-8){$\omega$}
\psline(18,-5)(18,5)(32,5)(32,-5)\rput(34,5){\small $1$}
\rput(16,-10){$-\frac{\pi}{L}$}\rput(32,-10){$\frac{\pi}{L}$}
}
\psline{->}(135,0)(135,-35)(105,-35)\psframe(85,-40)(105,-30)\rput(95,-35){$C(Z)$}\psline(85,-35)(65,-35)(65,-18)
\psline{->}(210,0)(235,0)\rput(243,1){$\hat{X}^{\dpcm}_n$}
\psline{<->}(25,-5)(25,-18)(105,-18)(105,-5)
}
\end{pspicture}
\end{center}
\caption{The test channel corresponding to the DPCM architecture, with the DPCM quantizer replaced by an AWGN channel. The input is assumed to be over-sampled at $L$ times the Nyquist rate.}
\label{fig:DPCMtest}
\end{figure*}
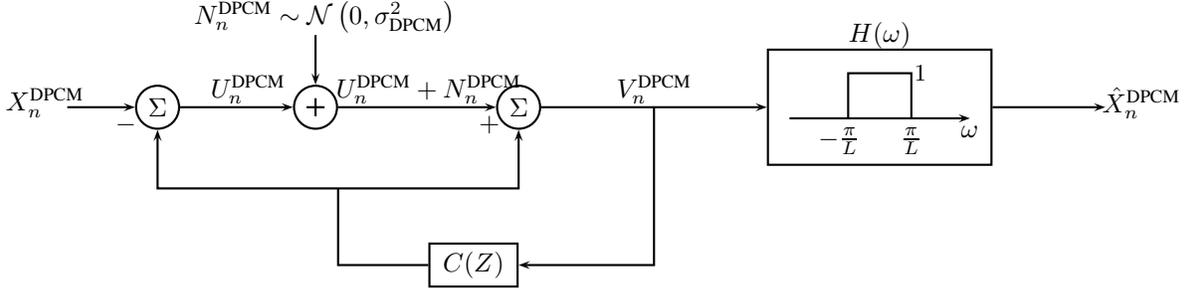 
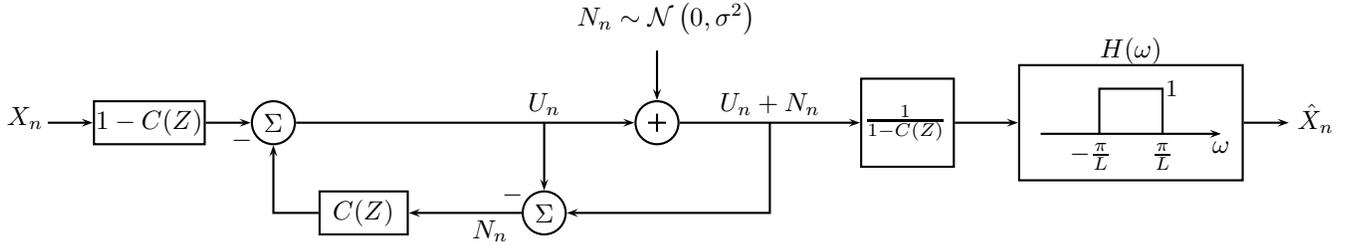
\begin{figure*}[]
\begin{center}
\psset{unit=0.6mm}
\begin{pspicture}(-30,0)(250,55)
\rput(-30,30){
\rput(0,1){$X_n$}\psline{->}(5,0)(15,0)\psframe(15,-5)(40,5)\rput(27.5,0){$1-C(Z)$}
}

\rput(0,30){
\psline{->}(10,0)(20,0)\pscircle(25,0){5}\rput(25,0){$\Sigma$}\rput(18,-4){$-$}\psline{->}(30,0)(105,0)\rput(85,4){$U_n$}
\pscircle(110,0){5}\psline(108,0)(112,0)\psline(110,-2)(110,2)\psline{->}(110,16)(110,5)
\rput(112,23){$N_n\sim\mathcal{N}\left(0,\sigma^2\right)$}%{$N_n\sim\mathcal{N}\left(0,\frac{D}{\frac{1}{2\pi}\int_{-\pi/L}^{\pi/L}|1-C(\omega)|^2 d\omega}\right)$}
\psline{->}(115,0)(155,0)\rput(135,4){$U_n+N_n$}
\psframe(155,-10)(176,10)\rput(165,0){$\frac{1}{1-C(Z)}$}\psline{->}(176,0)(190,0)
\psframe(190,-13)(240,13)\rput(215,16){$H(\omega)$}
\psline{->}(135,0)(135,-20)(90,-20)\pscircle(85,-20){5}\rput(85,-20){$\Sigma$}\rput(78,-16){$-$}\psline{->}(85,0)(85,-15)
\psline{->}(80,-20)(55,-20)\rput(73,-24){$N_n$}
\psframe(35,-25)(55,-15)\rput(45,-20){$C(Z)$}\psline{->}(35,-20)(25,-20)(25,-5)

\rput(190,2.5){
\psline{->}(5,-5)(45,-5)\rput(45,-8){$\omega$}
\psline(18,-5)(18,5)(32,5)(32,-5)\rput(34,5){\small $1$}
\rput(16,-10){$-\frac{\pi}{L}$}\rput(32,-10){$\frac{\pi}{L}$}
}
\psline{->}(240,0)(250,0)\rput(256,1){$\hat{X}_n$}
}
\end{pspicture}
\end{center}
\caption{A test channel corresponding to the sigma-delta modulation with pre-filter $1-C(Z)$ and post filter $\tfrac{1}{1-C(Z)}$. This test-channel is equivalent to that from Figure~\ref{fig:DPCMtest}. However, the pre-filter makes this architecture unattractive for data converters.}
\label{fig:PrePotsSigDelTest}
\end{figure*} 

\section{Main Result}
\label{sec:mainresult}

We begin by introducing some basic notation that will be used in the sequel. For a discrete signal $\{c_n\}$, the $Z$-transform is defined as
\begin{align}
C(Z)\triangleq\sum_{n=-\infty}^{\infty}c_n Z^{-n},\nonumber
\end{align}
and the Fourier transform as
\begin{align}
C(\omega)\triangleq C(Z)|_{Z=e^{j\omega}}=\sum_{n=-\infty}^{\infty}c_n e^{-j\omega n}.\nonumber
\end{align}
For a discrete (real) stationary process $\{X_n\}$ with zero-mean and autocorrelation function $R_{X}[k]\triangleq\mathbb{E}(X_{n+k} X_n)$ we define the power-spectral density (PSD) as the Fourier transform of the autocorrelation function
\begin{align}
S_{X}(\omega)\triangleq \sum_{k=-\infty}^{\infty}R_X[k] e^{-j\omega k}.\nonumber
\end{align}
The PSD of a continuous stationary process is defined in an analogous manner.

Assume $X^{\sd}(t)$ is a continuous stationary band-limited Gaussian process with zero mean and variance $\sigma^2_X$, whose PSD is zero for all frequencies $|f|> f_{\text{max}}$, but otherwise unknown. The Nyquist sampling rate for this process is $2 f_{\text{max}}$ samples per second. Since our focus here is on quantization of over-sampled signals, we assume that $X^{\sd}(t)$ is sampled uniformly with rate of $2L f_{\text{max}}$  samples per second for some $L>1$. The obtained sampled process $\{X^{\sd}_n\}$ is therefore a discrete stationary Gaussian process with zero mean and variance $\sigma_X^2$ whose PSD is zero for all $\omega\notin [-\pi/L,\pi/L]$, but otherwise unknown. Our goal is to characterize the rate-distortion trade-off obtained by a sigma-delta modulator, modeled as the test channel from Figure~\ref{fig:SigDelTest}, whose input is $\{X^{\sd}_n\}$. To that end, we establish an equivalence between the performance obtained by this test channel for any stationary band-limited Gaussian process with variance $\sigma^2_X$ and the performance obtained by the test channel from Figure~\ref{fig:DPCMtest}, which models a DPCM compression system, for a stationary \emph{flat} band-limited Gaussian process with variance $\sigma^2_X$. The performance of the latter is now well understood~\cite{zke08}, and, as we shall show, can be translated to a simple characterization of the performance of sigma-delta modulation.

First, we recall the derivation of the distortions attained by the test channels from Figure~\ref{fig:SigDelTest} and Figure~\ref{fig:DPCMtest}, and the scalar mutual information $I(U^{\sd}_n;U^{\sd}_n+N^{\sd}_n)$ and $I(U^{\dpcm}_n;U^{\dpcm}_n+N^{\dpcm}_n)$ between the input and output of the additive white Gaussian noise (AWGN) channels embedded within the two test channels.

The test channels in Figure~\ref{fig:SigDelTest} and Figure~\ref{fig:DPCMtest} do not immediately induce an output distribution from which a random quantization codebook with rate $I(U_n;U_n+N_n)$ and MSE distortion $D$ can be drawn. The reason for this is the sequential nature of the compression, which seems to conflict with the need of using high-dimensional quantizers, as required for attaining a quantization error distributed as $N_n$ with compression rate $I(U_n;U_n+N_n)$. Fortunately, this difficulty, which is also present in decision--feedback equalization for intersymbol interference channels, can be overcome with the help of an interleaver~\cite{gv05,zke08,oz09}. Thus, the scalar mutual information $I(U_n;U_n+N_n)$ can indeed be interpreted as the compression rate needed to achieve the distortion attained by the test channels in Figure~\ref{fig:SigDelTest} and Figure~\ref{fig:DPCMtest}. We elaborate further about this in subsection~\ref{subsec:vecquantization}. Moreover, in Section~\ref{sec:scalarquant} we show that $I(U_n;U_n+N_n)$ is closely related to the required quantization rate in a sigma-delta modulator that applies a \emph{uniform scalar quantizer} of finite support.

We begin with the test channel in Figure~\ref{fig:SigDelTest}, that corresponds to sigma-delta modulation, with the sigma-delta quantizer replaced by an AWGN channel with zero mean and variance $\sigma^2_{\Sigma\Delta}$. The filter $C(Z)$ is assumed to be strictly causal.

\vspace{1mm}

\begin{proposition}
For any Gaussian stationary process $\{X^{\sd}_n\}$ with variance $\sigma_X^2$ whose PSD is zero for all $\omega\notin [-\pi/L,\pi/L]$, the test channel from Figure~\ref{fig:SigDelTest} achieves MSE distortion
\begin{align}
D=\sigma^2_{\Sigma\Delta}\cdot\frac{1}{2\pi}\int_{-\pi/L}^{\pi/L}|1-C(\omega)|^2 d\omega,\nonumber
\end{align}
and its scalar mutual information satisfies\footnote{All logarithms in this paper are taken to base $2$.}
\begin{align}
I(U^{\sd}_n;U^{\sd}_n&+N^{\sd}_n)\nonumber\\
&=\frac{1}{2}\log\left(1+\frac{1}{2\pi}\int_{-\pi}^{\pi}|C(\omega)|^2 d\omega+\frac{\sigma^2_X}{\sigma^2_{\Sigma\Delta}}\right).\nonumber
\end{align}
\label{prop:SigDelRD}
\end{proposition}

\vspace{1mm}

\begin{proof}
From Figure~\ref{fig:SigDelTest}, we have that
\begin{align}
U^{\sd}_n&=X^{\sd}_n-c_n*N^{\sd}_n\label{UnSigDel},
\end{align}
and therefore
\begin{align}
U^{\sd}_n&+N^{\sd}_n=X^{\sd}_n+(\delta_n-c_n)*N^{\sd}_n,\nonumber
\end{align}
where $\delta_n$ is the discrete identity filter. Using the fact that $\{X^{\sd}_n\}$ is a low-pass process, passing it through the filter $H(\omega)$ has no effect, and hence
\begin{align}
\hat{X}^{\sd}_n&=h_n*(U^{\sd}_n+N^{\sd}_n)\nonumber\\
&=X^{\sd}_n+h_n*(\delta_n-c_n)*N^{\sd}_n.\nonumber
\end{align}
The MSE distortion attained by the test channel from Figure~\ref{fig:SigDelTest} is therefore
\begin{align}
D=\mathbb{E}(X^{\sd}_n-\hat{X}^{\sd}_n)^2=\sigma^2_{\Sigma\Delta}\cdot\frac{1}{2\pi}\int_{-\pi/L}^{\pi/L}|1-C(\omega)|^2 d\omega.\nonumber
\end{align}
The scalar mutual information between the ``quantizer's'' input $U^{\sd}_n$ and output $U^{\sd}_n+N^{\sd}_n$ is given by
\begin{align}
I(U^{\sd}_n;U^{\sd}_n+N^{\sd}_n)&=h(U^{\sd}_n+N^{\sd}_n)-h(N^{\sd}_n)\label{Nnindep}\\
&=\frac{1}{2}\log\left(1+\frac{\mathbb{E}(U^{\sd}_n)^2}{\sigma^2_{\Sigma\Delta}}\right)\label{Iineqsigdel},
\end{align}
where~\eqref{Nnindep}, as well as~\eqref{Iineqsigdel}, follow from the statistical independence of $N^{\sd}_n$ and $U^{\sd}_n$. Using~\eqref{UnSigDel}, the variance of $U^{\sd}_n$ is
\begin{align}
&\mathbb{E}(U^{\sd}_n)^2=\sigma_X^2+\sigma^2_{\Sigma\Delta}\frac{1}{2\pi}\int_{-\pi}^{\pi}|C(\omega)|^2 d\omega.\label{UnSigDelVar}
\end{align}
Substituting~\eqref{UnSigDelVar} into~\eqref{Iineqsigdel} establishes the second part of the proposition.
\end{proof}

\vspace{1mm}

Next, we analyze the test channel in Figure~\ref{fig:DPCMtest}, that corresponds to DPCM compression with the DPCM quantizer replaced by an AWGN channel with zero mean and variance $\sigma^2_{\text{DPCM}}$. As in the test channel of Figure~\ref{fig:SigDelTest}, the filter $C(Z)$ is strictly causal. The distortion corresponding to this test channel, as well as $I(U^{\dpcm}_n;U^{\dpcm}_n+N^{\dpcm}_n)$, were already found in~\cite[Theorem 1]{zke08} for the special case where $C(Z)$ is the optimal MMSE infinite length prediction filter of $X^{\dpcm}_n$ from all past samples of the process $\{X^{\dpcm}_n+N^{\dpcm}_n\}$. The following straightforward proposition characterizes the rate and distortion for any choice of the causal filter $C(Z)$ and any value of $\sigma^2_{\text{DPCM}}$.

\vspace{1mm}

\begin{proposition}
For a Gaussian stationary process $\{X^{\dpcm}_n\}$ with variance $\sigma_X^2$ and PSD
\begin{align}
S^{\dpcm}_X(\omega)=\begin{cases}
L\sigma_X^2 & \text{for } |\omega|\leq \pi/L \\
0 & \text{for } \pi/L<|\omega|<\pi
\end{cases},\label{Xpsd}
\end{align}
the test channel from Figure~\ref{fig:DPCMtest} achieves MSE distortion
\begin{align}
D=\frac{\sigma^2_{\text{DPCM}}}{L},\nonumber
\end{align}
and its scalar mutual information satisfies
\begin{align}
I(U^{\dpcm}_n;U^{\dpcm}_n&+N^{\dpcm}_n)=\frac{1}{2}\log\bigg(1+\frac{1}{2\pi}\int_{-\pi}^{\pi}|C(\omega)|^2 d\omega\nonumber\\
& \ \ \ \ \ \ \ \   +\frac{L\sigma^2_X}{\sigma^2_{\text{DPCM}}}\frac{1}{2\pi}\int_{-\pi/L}^{\pi/L}|1-C(\omega)|^2 d\omega\bigg).\nonumber
\end{align}
\label{prop:DPCMRD}
\end{proposition}

\vspace{1mm}

\begin{proof}
From Figure~\ref{fig:DPCMtest}, we have that
\begin{align}
U^{\dpcm}_n&=X^{\dpcm}_n-c_n*V^{\dpcm}_n\label{UnDPCMraw}\\
V^{\dpcm}_n&=U^{\dpcm}_n+N^{\dpcm}_n+c_n*V^{\dpcm}_n\label{UnDPCMrawb}
\end{align}
Substituting~\eqref{UnDPCMraw} in~\eqref{UnDPCMrawb} yields
\begin{align}
V^{\dpcm}_n&=X^{\dpcm}_n+N^{\dpcm}_n.\label{VnDPCM}
\end{align}
Using the fact that $\{X^{\dpcm}_n\}$ is a low-pass process, as before, we obtain
\begin{align}
\hat{X}^{\dpcm}_n&=h_n*(X^{\dpcm}_n+N^{\dpcm}_n)\nonumber\\
&=X^{\dpcm}_n+h_n*N^{\dpcm}_n.\label{LPFeqDPCM}
\end{align}
Since $\{N^{\dpcm}_n\}$ is AWGN with variance $\sigma^2_{\text{DPCM}}$, the variance of the filtered process $h_n*N^{\dpcm}_n$ is $\sigma^2_{\text{DPCM}}/L$. Thus,
\begin{align}
D=\mathbb{E}(X^{\dpcm}_n-\hat{X}^{\dpcm}_n)^2=\frac{\sigma^2_{\text{DPCM}}}{L}.\nonumber
\end{align}
As in the analysis of the test channel from Figure~\ref{fig:SigDelTest}, the scalar mutual information between $U^{\dpcm}_n$ and $U^{\dpcm}_n+N^{\dpcm}_n$ is given by
\begin{align}
I(U^{\dpcm}_n;U^{\dpcm}_n+N^{\dpcm}_n)=\frac{1}{2}\log\left(1+\frac{\mathbb{E}(U^{\dpcm}_n)^2}{\sigma^2_{\text{DPCM}}}\right)\label{Idpcm}.
\end{align}
Now, substituting~\eqref{VnDPCM} in~\eqref{UnDPCMraw} gives
\begin{align}
U^{\dpcm}_n&=(\delta_n-c_n)*X^{\dpcm}_n-c_n*N^{\dpcm}_n,\nonumber
\end{align}
and the variance of $U_n$ is therefore
\begin{align}
&\mathbb{E}(U^{\dpcm}_n)^2=\frac{1}{2\pi}\int_{-\pi}^{\pi}S^{\dpcm}_X(\omega)|1-C(\omega)|^2d\omega\nonumber\\
& \ \ \ \ \ \ \ \ \ \ \ \ \ \ +\frac{1}{2\pi}\int_{-\pi}^{\pi}S^{\dpcm}_N(\omega)|C(\omega)|^2 d\omega\nonumber\\
&=\frac{L\sigma_X^2}{2\pi}\int_{-\pi/L}^{\pi/L}|1-C(\omega)|^2 d\omega+\frac{\sigma^2_{\text{DPCM}}}{2\pi}\int_{-\pi}^{\pi}|C(\omega)|^2 d\omega.\label{UnVar}
\end{align}
Substituting~\eqref{UnVar} into~\eqref{Idpcm} establishes the second part of the proposition.
\end{proof}

\begin{remark}
In propositions~\ref{prop:SigDelRD} and~\ref{prop:DPCMRD} we derived the \emph{scalar} mutual information between the input and output of the AWGN test channels embedded in Figures~\ref{fig:SigDelTest} and~\ref{fig:DPCMtest}, respectively. As will become clear in Section~\ref{sec:scalarquant}, the scalar mutual information is closely related to the required quantization rate when a scalar memoryless quantizer is used within the sigma-delta or DPCM modulator. In~\cite{zke08,oz09}, the directed information was shown to be related to the required quantization rate when the quantizer is followed by an entropy coder. Here, we do not consider applying entropy coding to the quantizer's output as we require that the designed modulator be robust to the statistics of the input process, whereas entropy coding is very sensitive to the process statistics. Moreover, if the design of an A/D (or D/A) is considered, the appropriate merit for the modulator's complexity is the number of quantization levels within the scalar quantizer, which are not reduced by incorporating an entropy coder.
\end{remark}

\vspace{1mm}

Our main result now follows immediately from Propositions~\ref{prop:SigDelRD} and~\ref{prop:DPCMRD}.

\vspace{1mm}

\begin{theorem}
Let $\{X^{\sd}_n\}$ be any Gaussian stationary process with variance $\sigma_X^2$ whose PSD is zero for all $\omega\notin [-\pi/L,\pi/L]$, let $\{X^{\dpcm}_n\}$ be a flat low-pass Gaussian stationary process with PSD as in~\eqref{Xpsd}, and let $C(Z)$ be a strictly causal filter. The test channel from Figure~\ref{fig:SigDelTest} with
\begin{align}
\sigma^2_{\Sigma\Delta}=\frac{D}{\frac{1}{2\pi}\int_{-\pi/L}^{\pi/L}|1-C(\omega)|^2 d\omega},\nonumber
\end{align}
and the test channel from Figure~\ref{fig:DPCMtest} with
\begin{align}
\sigma^2_{\text{DPCM}}=L\cdot D,\nonumber
\end{align}
both achieve MSE distortion $D$ and their scalar mutual information satisfy
\begin{align}
I(U^{\sd}_n&;U^{\sd}_n+N^{\sd}_n)=I(U^{\dpcm}_n;U^{\dpcm}_n+N^{\dpcm}_n)\nonumber\\
&=\frac{1}{2}\log\bigg(1+\frac{1}{2\pi}\int_{-\pi}^{\pi}|C(\omega)|^2 d\omega\nonumber\\
& \ \ \ \ \ \ \ \ \ \ \ \ +\frac{\sigma^2_X}{D}\frac{1}{2\pi}\int_{-\pi/L}^{\pi/L}|1-C(\omega)|^2 d\omega\bigg).\nonumber
\end{align}
\label{thm:SigDelDPCMeq}
\end{theorem}

\vspace{1mm}

This theorem indicates that for any stationary band-limited Gaussian process with variance $\sigma^2_X$, the sigma-delta test channel from Figure~\ref{fig:SigDelTest} achieves exactly the same rate-distortion trade-off as that of the DPCM test channel from Figure~\ref{fig:DPCMtest} with a stationary flat band-limited Gaussian input with the same variance, provided that the AWGN variances are scaled according to~\eqref{sigdelnoise}. Thus, Theorem~\ref{thm:SigDelDPCMeq} provides a unified framework for analyzing the performance of sigma-delta modulation and DPCM. A great advantage offered by such a unified framework, is that any result known for DPCM can be translated to a corresponding result for sigma-delta modulation, and vice versa. Theorems~\ref{thm:optfilt} and Corollary~\ref{thm:sigdeloptimality} below constitute two important examples of such results.

\vspace{1mm}

\begin{theorem}
Let $\{X^{\sd}_n\}$ be a Gaussian stationary process with variance $\sigma_X^2$ whose PSD is zero for all $\omega\notin [-\pi/L,\pi/L]$  and let $\mathcal{C}$ be a family of strictly causal filters. Define the ``virtual'' process $\{S_n\}$ as a Gaussian stationary process with PSD as in~\eqref{Xpsd}, and the ``virtual'' process $\{W_n\}$ as a Gaussian i.i.d. random process statistically independent of $\{S_n\}$ with variance $L\cdot D$, $D>0$. Let
\begin{align}
\sigma^{*2}_D&=\min_{C(Z)\in\mathcal{C}}\mathbb{E}\left(S_n-c_n*(S_n+W_n)\right)^2\nonumber\\
C_D^*(Z)&=\argmin_{C(Z)\in\mathcal{C}}\mathbb{E}\left(S_n-c_n*(S_n+W_n)\right)^2.\nonumber
\end{align}
If the filter $C(Z)$ in the sigma-delta test channel from Figure~\ref{fig:SigDelTest} belongs to $\mathcal{C}$ and the MSE distortion attained by this test channel is $D$, then
\begin{align}
I(U^{\sd}_n;U^{\sd}_n+N^{\sd}_n)\geq&\frac{1}{2}\log\left(1+\frac{\sigma^{*2}_D}{L\cdot D}\right),\label{Idmin}
\end{align}
with equality if $C(Z)=C_D^*(Z)$.
\label{thm:optfilt}
\end{theorem}

\vspace{2mm}

Theorem~\ref{thm:optfilt} states that for a target distortion $D$, the sigma-delta filter which minimizes the required compression rate is the optimal linear time-invariant MMSE estimator, within the class of constraints $\mathcal{C}$, for $S_n$ from the past of the noisy process $\{S_n+W_n\}$. For example, if $\mathcal{C}$ consists of all strictly causal finite-impulse response (FIR) filters of length $p$, the optimal filter $C(Z)$ is the optimal predictor of $S_n$ from the samples $\{S_{n-1}+W_{n-1},\ldots,S_{n-p}+W_{n-p}\}$, which can be easily calculated in closed-form.

The optimal sigma-delta filter design problem was studied by several authors, under various assumptions~\cite{jayantnoll,th78,noll78,gc87,sc62,dsqg08,do12}. However, to the best of our knowledge, the simple expression from Theorem~\ref{thm:optfilt} for the optimal filter as the optimal predictor of $S_n$ from the past of $\{S_n+W_n\}$ is novel. The references most relevant to Theorem~\ref{thm:optfilt}, are perhaps~\cite{sc62} and~\cite{dsqg08,do12}. In~\cite{sc62}, Spang and Schultheiss formulated an optimization problem for finding the best FIR filter with $p$ coefficients in a sigma-delta modulator with a scalar quantizer, under a fixed overload probability. Their optimization problem can be solved numerically, but no closed form solution was given.
In~\cite{dsqg08} and~\cite{do12} the design of an optimal \emph{unconstrained} sigma-delta filter was studied, under the assumption of a fixed scalar quantizer which can only be scaled in order to control the overload probability. Equations that characterize the optimal filter were derived. However, the obtained expressions usually yield filters with an infinite number of taps, and do not provide the solution to the constrained problem. It is also worth mentioning that for the case of a stationary Gaussian process $\{X_n\}$ with $L=1$ (sampling at the Nyquist rate) and \emph{known PSD}, the optimal infinite length filter under the assumption of high-resolution quantization is known to equal the optimal prediction filter of $X_n$ from its (clean) past~\cite{noll78}. As already mentioned in the introduction, the high-resolution assumption never holds when $L>1$ and therefore this result is inapplicable for over-sampled signals.

\vspace{2mm}

\begin{proof}[Proof of Theorem~\ref{thm:optfilt}]
By Proposition~\ref{prop:SigDelRD}, if the test channel from Figure~\ref{fig:SigDelTest} achieves MSE distortion $D$, we must have
\begin{align}
\sigma^2_{\Sigma\Delta}=\frac{D}{\frac{1}{2\pi}\int_{-\pi/L}^{\pi/L}|1-C(\omega)|^2 d\omega}.\nonumber
\end{align}
By Theorem~\ref{thm:SigDelDPCMeq}, the corresponding mutual information $I(U^{\sd}_n;U^{\sd}_n+N^{\sd}_n)$ is equal to the mutual information $I(U^{\dpcm}_n;U^{\dpcm}_n+N^{\dpcm}_n)$ in the DPCM test channel from Figure~\ref{fig:DPCMtest} with $X^{\dpcm}_n=S_n$, $N^{\dpcm}_n=W_n$ and $\sigma^2_{\text{DPCM}}=L\cdot D$. Thus,
\begin{align}
&I(U^{\sd}_n;U^{\sd}_n+N^{\sd}_n)=I(U^{\dpcm}_n;U^{\dpcm}_n+N^{\dpcm}_n)\nonumber\\
&=\frac{1}{2}\log\left(1+\frac{\mathbb{E}\left(S_n-c_n*(S_n+W_n)\right)^2}{L\cdot D}\right),\label{Idpcmwithcn}
\end{align}
where we have used~\eqref{UnDPCMraw},~\eqref{VnDPCM}, and~\eqref{Idpcm}, to arrive at~\eqref{Idpcmwithcn}.
It follows that among all filters in $\mathcal{C}$, the  filter that minimizes~\eqref{Idpcmwithcn} is $C_D^*(Z)$, and that it attains~\eqref{Idmin} with equality.
\end{proof}

\vspace{2mm}

It is interesting to note~\cite{zke08} that since $\{W_n\}$ is an i.i.d. process with variance $L\cdot D$ and $C(Z)$ is strictly causal, the mutual information~\eqref{Idpcmwithcn} can also be written as
\begin{align}
I(U^{\sd}_n&;U^{\sd}_n+N^{\sd}_n)\nonumber\\
&=\frac{1}{2}\log\left(\frac{\mathbb{E}\left(S_n+W_n-c_n*(S_n+W_n)\right)^2}{L\cdot D}\right).\label{Idpcmwithcn2}
\end{align}
Thus, the optimal predictor of $S_n$ from the past of $\{S_n+W_n\}$ is identical to the optimal predictor of $S_n+W_n$ from its past samples. When $C(Z)$ is taken as the (unique) infinite order optimal one-step prediction filter of $S_n+W_n$ from its past samples, the prediction error variance is the entropy power of the process $\{S_n+W_n\}$~\cite{berger71}, which equals
\begin{align}
2^{\frac{1}{2\pi}\int_{-\pi}^{\pi}\log\left(S_S(\omega)+L\cdot D \right)d\omega}=(L\cdot D)\left(1+\frac{\sigma_X^2}{D}\right)^{1/L}.\label{entpow}
\end{align}
Moreover, the infinite order prediction error
\begin{align}
E^{\text{pred}}_n\triangleq S_n+W_n-c_n*(S_n+W_n)\nonumber
\end{align}
is in this case a white process. This, together with~\eqref{entpow} implies that for the optimal unconstrained sigma-delta filter $C(Z)$ we must have
\begin{align}
S_{E^{\text{pred}}}(\omega)&\triangleq |1-C(\omega)|^2\left(L\cdot D+S_S(\omega)\right)\nonumber\\
&=(L\cdot D)\left(1+\frac{\sigma_X^2}{D}\right)^{1/L}, \ \forall\omega\in[-\pi,\pi)\label{infinitefiltspec}
\end{align}
Combining~\eqref{Idpcmwithcn2},~\eqref{entpow}, and~\eqref{infinitefiltspec} yields the following corollary.

\vspace{2mm}

\begin{corollary}
Let $\{X^{\sd}_n\}$ be a Gaussian stationary process with variance $\sigma_X^2$ whose PSD is zero for all $\omega\notin [-\pi/L,\pi/L]$. If the test channel from Figure~\ref{fig:SigDelTest} attains MSE distortion $D$, then
\begin{align}
I(U^{\sd}_n;U^{\sd}_n+N^{\sd}_n)\geq \frac{1}{2L}\log\left(1+\frac{\sigma_X^2}{D}\right).\label{unbiasedrd}
\end{align}
with equality if and only if $C(Z)$ is a strictly causal filter satisfying
\begin{align}
|1-C(\omega)|^2=\begin{cases}
\left(1+\frac{\sigma_X^2}{D}\right)^{-(L-1)/L} & \omega\in[-\frac{\pi}{L},\frac{\pi}{L}]\\
\left(1+\frac{\sigma_X^2}{D}\right)^{1/L} & \omega\notin[-\frac{\pi}{L},\frac{\pi}{L}],
\end{cases}
\label{optinfinitefilt}
\end{align}
and
\begin{align}
\sigma^2_{\Sigma\Delta}=\frac{D}{\frac{1}{2\pi}\int_{-\pi/L}^{\pi/L}|1-C(\omega)|^2 d\omega}=\frac{L\cdot D}{\left(1+\frac{\sigma_X^2}{D}\right)^{-(L-1)/L}}.\nonumber
\end{align}
\label{thm:sigdeloptimality}
\end{corollary}

\vspace{1mm}

\begin{remark}
Note that the existence of a strictly causal filter $C(Z)$ which satisfies~\eqref{optinfinitefilt} is guaranteed by Wiener's spectral-factorization theory~\cite{berger71} due to the readily verified fact that
\begin{align}
2^{\frac{1}{2\pi}\int_{-\pi}^{\pi}\log|1-C(\omega)|^2 d\omega}=1.\nonumber
\end{align}
The optimal filter induces a two-level frequency response for $|1-C(\omega)|^2$.
In~\cite{oz09} {\O}stergaard and Zamir used sigma-delta modulation to attain the optimal multiple-description rate-distortion region. Interestingly, the optimal filter $C(Z)$ in their scheme also induced a two-level response for $|1-C(\omega)|^2$. We also note that the optimality of the unconstrained filter specified by~\eqref{optinfinitefilt} can be deduced as a special case of~\cite[Section IV]{dsqg08}.
\label{rem:filtexist}
\end{remark}

\vspace{1mm}

\begin{remark}
Note that for the optimal unconstrained filter $C(Z)$ specified by~\eqref{optinfinitefilt}, the pre- and post-filters from Figure~\ref{fig:PrePotsSigDelTest} have no effect as long as the PSD of the input signal $\{X^{\sd}_n\}$ is zero for all $\omega\notin [-\pi/L,\pi/L]$. However, filters with a finite number of taps will never incur a flat frequency response in the interval $[-\pi/L,\pi/L]$, and for such filters the systems from Figure~\ref{fig:SigDelTest} and Figure~\ref{fig:PrePotsSigDelTest} will not be equivalent.
\end{remark}

\vspace{1mm}

\begin{remark}
The output of the test channel from Figure~\ref{fig:SigDelTest} (as well as that from Figure~\ref{fig:DPCMtest}) is of the form $\hat{X}^{\sd}_n=X^{\sd}_n+E^{\sd}_n$, where $E^{\sd}_n$ has zero mean and variance $D$, and is statistically independent of $X^{\sd}_n$. This estimate can be further improved by applying scalar MMSE estimation for $X^{\sd}_n$ from $\hat{X}^{\sd}_n$. This boils down to producing the estimate $\hat{\tilde{X}}^{\sd}_n=\alpha\hat{X}^{\sd}_n$, where
\begin{align}
\alpha=\frac{\sigma_X^2}{\sigma_X^2+D}.\nonumber
\end{align}
Consequently, the obtained MSE distortion is reduced to
\begin{align}
\tilde{D}=\mathbb{E}(X^{\sd}_n-\alpha\hat{X}^{\sd}_n)^2=\frac{\sigma_X^2\cdot D}{\sigma_X^2+D}.\nonumber
\end{align}
It is straightforward to verify~\cite{ramibook} that with this improvement, the sigma-delta test channel from Figure~\ref{fig:SigDelTest} with $C(Z)$ and $\sigma^2_{\Sigma\Delta}$ as specified in Corollary~\ref{thm:sigdeloptimality} attains
\begin{align}
I(U^{\sd}_n;U^{\sd}_n+V^{\sd}_n)=\frac{1}{2L}\log\left(\frac{\sigma_X^2}{\tilde{D}} \right),\nonumber
\end{align}
which is the optimal rate-distortion function for a stationary Gaussian source $\{X^{\sd}_n\}$ with PSD as in~\eqref{Xpsd}. It follows that the sigma-delta test channel from Figure~\ref{fig:SigDelTest} with $C(Z)$ and $\sigma^2_{\Sigma\Delta}$ as specified in Corollary~\ref{thm:sigdeloptimality} is minimax optimal for the class of all stationary Gaussian sources with variance $\sigma^2_X$ and PSD that equals zero for all $\omega\notin [-\pi/L,\pi/L]$, i.e., no other system can achieve MSE distortion $\tilde{D}$ with a smaller compression rate, universally for all sources in this class.
\label{rem:RdOptimality}
\end{remark}

\subsection{Extension to Frequency-Weighted Mean Squared Error Distortion}
\label{subsec:fwmse}

In many applications, higher values of distortion are acceptable in certain frequency bands while smaller distortion is permitted in other bands. The MSE distortion measure is inadequate for such scenarios, and a commonly used distortion measure, that (partially) captures such perceptual effects, is the frequency-weighted mean squared error (FWMSE) criterion. Under this criterion, the distortion is measured as
\begin{align}
D_{\text{FWMSE}}\triangleq\frac{1}{2\pi}\int_{-\pi}^{\pi} P(\omega)S_E(\omega) d\omega,
\end{align}
where $P(\omega)$ is a non-negative weight function, and $S_E(\omega)$ is the PSD of the error process $E_n\triangleq X^{\sd}_n-\hat{X}^{\sd}_n$. Note that for $P(\omega)=1, \forall \omega\in[-\pi,\pi)$, the FWMSE criterion reduces to the MSE one. The next theorem shows that the constrained optimal sigma-delta filter under the FWMSE criterion is the optimal constrained prediction filter of a noisy process defined according to the weight function $P(\omega)$.

\vspace{1mm}

\begin{theorem}
Let $\{X^{\sd}_n\}$ be a Gaussian stationary process with variance $\sigma_X^2$ whose PSD is zero for all $\omega\notin [-\pi/L,\pi/L]$, $P(\omega)$ a weighting function which forms a valid PSD, and $\mathcal{C}$ a family of strictly causal filters.
Define the ``virtual'' process  $\{S_n\}$ as a Gaussian stationary process with PSD
\begin{align}
S^{\text{FWMSE}}_X(\omega)=\begin{cases}
L\sigma_X^2 P(\omega) & \text{for } |\omega|\leq \pi/L \\
0 & \text{for } \pi/L<|\omega|<\pi
\end{cases},\label{XfwmsePSD}
\end{align}
and the ``virtual'' process $\{W_n\}$ as a Gaussian i.i.d. random process statistically independent of $\{S_n\}$ with variance \mbox{$L\cdot D_{\text{FWMSE}}$}, $D_{\text{FWMSE}}>0$. Let
\begin{align}
\sigma^{*2}_{D_{\text{FWMSE}}}&=\min_{C(Z)\in\mathcal{C}}\mathbb{E}\left(S_n-c_n*(S_n+W_n)\right)^2\nonumber\\
C_{D_{\text{FWMSE}}}^*(Z)&=\argmin_{C(Z)\in\mathcal{C}}\mathbb{E}\left(S_n-c_n*(S_n+W_n)\right)^2.\nonumber
\end{align}
If the filter $C(Z)$ in the sigma-delta test channel from Figure~\ref{fig:SigDelTest} belongs to $\mathcal{C}$ and the FWMSE distortion w.r.t. $P(\omega)$ attained by this test channel is $D_{\text{FWMSE}}$, then
\begin{align}
I(U^{\sd}_n;U^{\sd}_n+N^{\sd}_n)\geq&\frac{1}{2}\log\left(1+\frac{\sigma^{*2}_{D_{\text{FWMSE}}}}{L\cdot {D_{\text{FWMSE}}}}\right),\nonumber
\end{align}
with equality if $C(Z)=C_{D_{\text{FWMSE}}}^*(Z)$.
\label{thm:optfiltFWMSE}
\end{theorem}

\vspace{1mm}

\begin{proof}[Sketch of proof:]
The proof is fairly similar to that of Theorem~\ref{thm:optfilt}. Thus, for brevity, we omit the full proof and only highlight its main steps:
\begin{itemize}
\item Repeat the derivation of Proposition~\ref{prop:SigDelRD} where now the MSE distortion is replaced by FWMSE distortion. Note that this has no effect on $I(U^{\sd}_n;U^{\sd}_n+N^{\sd}_n)$.
\item Repeat the derivation of Proposition~\ref{prop:DPCMRD} where the PSD of the input process is~\eqref{XfwmsePSD}, rather than~\eqref{Xpsd}. Note that this changes $I(U^{\dpcm}_n;U^{\dpcm}_n+N^{\dpcm}_n)$, but has no effect on the attained distortion.
\item It follows that the DPCM test channel for the process $\{S_n\}$ under MSE distortion is equivalent to the sigma-delta test channel with input $\{X^{\sd}_n\}$ under FWMSE distortion, in the sense that in both channels if the attained distortion is $D_{\text{FWMSE}}$ (under the appropriate distortion measure), then
\begin{align}
&I(U^{\sd}_n;U^{\sd}_n+N^{\sd}_n)=I(U^{\dpcm}_n;U^{\dpcm}_n+N^{\dpcm}_n)\nonumber\\
&=\frac{1}{2}\log\left(1+\frac{\mathbb{E}\left(S_n-c_n*(S_n+W_n)\right)^2}{L\cdot D_{\text{FWMSE}}}\right).\nonumber
\end{align}
\end{itemize}
\end{proof}

\subsection{Sigma-Delta Modulation with an Interleaved Vector Quantizer}
\label{subsec:vecquantization}
The goal of this short subsection is to give the test channel from Figure~\ref{fig:SigDelTest} an operational meaning, i.e., to show how the AWGN from the figure can be replaced with a lossy source code of rate $R=I(U^{\sd}_n;U^{\sd}_n+N^{\sd}_n)$ whose incurred quantization noise is distributed as $N_n^{\sd}$. As already mentioned, the key idea is to use an interleaver~\cite{gv05,zke08,oz09}, as we now recall.

Assume that $\{X_n^{\sd}\}$, the input process to the sigma-delta modulator, has a decaying memory, such that $X_n^{\sd}$ is essentially independent of all samples of sufficiently distant sampling times. In order to compress an $N$-dimensional vector
\begin{align}
\bx^{\sd}=[X_1^{\sd},\ldots,X_N^{\sd}],\nonumber
\end{align}
containing $N$ consecutive samples of the process $\{X_n^{\sd}\}$, we first split it into $K$ vectors
\begin{align}
\bx^{\sd}_k=[X_{(k-1)M+1}^{\sd},\ldots,X_{kM}^{\sd}], \ \ \  k=1,\ldots, K, \nonumber
\end{align}
where $M\triangleq N/K$.
Now, we can apply $K$ parallel sigma-delta modulators, one for each such vector, where the only coupling between the $K$ parallel systems is through the quantization step, which is applied jointly on all of them, as depicted in Figure~\ref{fig:parallelsigdel}.
\begin{figure*}[]
\begin{center}
\psset{unit=0.6mm}
\begin{pspicture}(0,-42)(250,40)
\rput(0,30){
\rput(0,1){$\bx^{\sd}_1$}\psline{->}(5,0)(20,0)\pscircle(25,0){5}\rput(25,0){$\Sigma$}\rput(18,-4){$-$}\psline{->}(30,0)(105,0)\rput(85,4){$U^{\sd}_{1,n}$}
\psframe(105,-5)(121,5)\rput(113,0){\small$Q_1(\cdot)$}
\psline{->}(121,0)(160,0)\rput(140,4){$U^{\sd}_{1,n}+N^{\sd}_{1,n}$}\psframe(160,-13)(210,13)\rput(185,16){$H(\omega)$}
\psline{->}(135,0)(135,-20)(90,-20)\pscircle(85,-20){5}\rput(85,-20){$\Sigma$}\rput(78,-16){$-$}\psline{->}(85,0)(85,-15)
\psline{->}(80,-20)(55,-20)\rput(73,-24){$N^{\sd}_{1,n}$}
\psframe(35,-25)(55,-15)\rput(45,-20){$C(Z)$}\psline{->}(35,-20)(25,-20)(25,-5)
\rput(160,2.5){
\psline{->}(5,-5)(45,-5)\rput(45,-8){$\omega$}
\psline(18,-5)(18,5)(32,5)(32,-5)\rput(34,5){\small $1$}
\rput(16,-10){$-\frac{\pi}{L}$}\rput(32,-10){$\frac{\pi}{L}$}
}
\psline{->}(210,0)(235,0)\rput(241,1){$\mathbf{\hat{x}}^{\sd}_1$}
}
\rput(112,0){\Large$\vdots$}

\rput(0,-18){
\rput(0,1){$\bx^{\sd}_M$}\psline{->}(5,0)(20,0)\pscircle(25,0){5}\rput(25,0){$\Sigma$}\rput(18,-4){$-$}\psline{->}(30,0)(105,0)\rput(85,4){$U^{\sd}_{K,n}$}
\psframe(105,-5)(121,5)\rput(113,0){\small$Q_K(\cdot)$}
\psline{->}(121,0)(160,0)\rput(140,4){$U^{\sd}_{K,n}+N^{\sd}_{K,n}$}\psframe(160,-13)(210,13)\rput(185,16){$H(\omega)$}
\psline{->}(135,0)(135,-20)(90,-20)\pscircle(85,-20){5}\rput(85,-20){$\Sigma$}\rput(78,-16){$-$}\psline{->}(85,0)(85,-15)
\psline{->}(80,-20)(55,-20)\rput(73,-24){$N^{\sd}_{K,n}$}
\psframe(35,-25)(55,-15)\rput(45,-20){$C(Z)$}\psline{->}(35,-20)(25,-20)(25,-5)
\rput(160,2.5){
\psline{->}(5,-5)(45,-5)\rput(45,-8){$\omega$}
\psline(18,-5)(18,5)(32,5)(32,-5)\rput(34,5){\small $1$}
\rput(16,-10){$-\frac{\pi}{L}$}\rput(32,-10){$\frac{\pi}{L}$}
}
\psline{->}(210,0)(235,0)\rput(241,1){$\mathbf{\hat{x}}^{\sd}_K$}
}

\psframe[linestyle=dashed](102,-27)(122.5,37)\rput(110,41){$\underline{Q}(\cdot)$}

\end{pspicture}
\end{center}
\caption{$K$ parallel sigma-delta modulators coupled by an $K$-dimensional quantizer $\underline{Q}(\cdot)$.}
\label{fig:parallelsigdel}
\end{figure*}
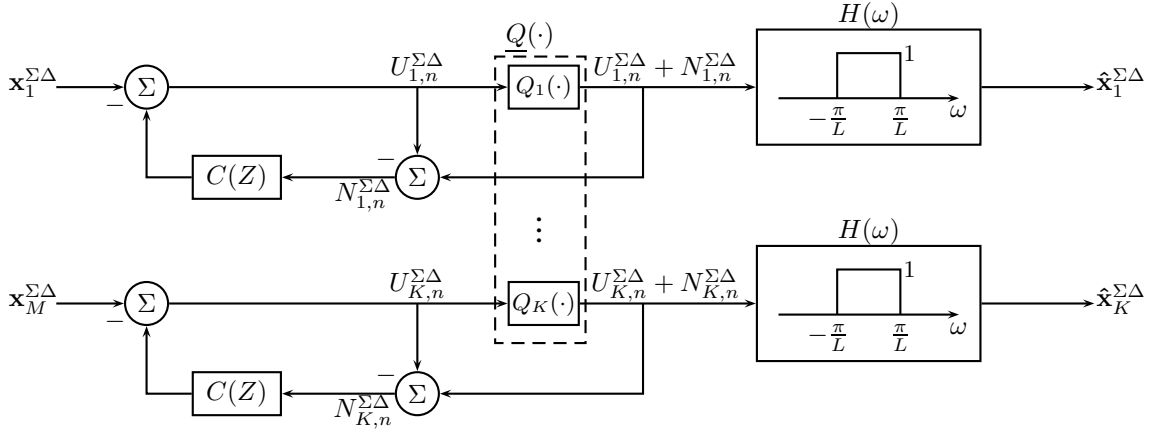 
By our assumption that $\{X_n^{\sd}\}$ has decaying memory, if $M$ is large enough the $K$ inputs that enter the quantizer $\underline{Q}(\cdot)=[Q_1(\cdot),\ldots,Q_K(\cdot)]$ are i.i.d. random variables distributed as $U^{\sd}_n$ from Figure~\ref{fig:SigDelTest}. For large enough $K$, standard rate-distortion arguments imply that there exists a vector quantizer with rate $I(U^{\sd}_n;U^{\sd}_n+N^{\sd}_n)$ that induces quantization noise distributed as $N^{\sd}_n$.

\section{Sigma-Delta Modulation with a Scalar Uniform Quantizer}
\label{sec:scalarquant}

The previous subsection showed how to replace the AWGN channel in Figure~\ref{fig:SigDelTest} with a vector quantizer whose rate is arbitrarily close to $R=I(U^{\sd}_n;U^{\sd}_n+N^{\sd}_n)$ and whose induced quantization noise is distributed as $N^{\sd}_n$. The inputs to the vector quantizer are vectors of i.i.d. Gaussian components. Thus, any ``off-the-shelf'' rate--distortion optimal vector quantizer for an i.i.d. Gaussian source can be used. The total sigma-delta compression system that is obtained is therefore simple in the sense that it only requires the vector quantizer to be good for quantizing an i.i.d. Gaussian source, which is a standard task, rather than requiring it to be a good quantizer for a band-limited Gaussian source.

However, the sigma-delta modulation architecture is mainly used for A/D and D/A conversion. In such applications, vector quantization is typically out of the question, and simple uniform scalar quantizers of finite support are used. For such quantizers, the quantization error is composed of two main factors~\cite{jayantnoll}: \emph{granular errors} that correspond to the quantization error in the case where the input signal falls within the quantizer's support, and \emph{overload errors} that correspond to the case where the input signal falls outside the quantizer's support. Due to the feedback loop, inherent to the sigma-delta modulator, errors of the latter kind, whose magnitude is not bounded, may have a disastrous effect as they jeopardize the system's stability. In order to avoid such errors, the support of the quantizer has to be chosen appropriately. As the support of the quantizer determines its rate for a given quantization resolution, the overload probability can be controlled by increasing the quantization rate.\footnote{As discussed in Section~\ref{subsec:contributions}, one can try to limit the effect of overload errors by placing various constraints on $C(Z)$. Here, we restrict attention to controlling the overload probability.}

We shall show that, given that overload errors did not occur, the quantization noise can be modeled as an additive noise. Thus, the test channel from Figure~\ref{fig:SigDelTest} accurately predicts the total distortion incurred by a sigma-delta A/D (or D/A) in this case. Moreover, the overload probability is a doubly exponentially decreasing function of $R-I(U^{\sd}_n;U^{\sd}_n+N^{\sd}_n)$, where $2^R$ are the number of levels in the scalar quantizer. Thus, fixing the desired overload error probability as $P_{ol}$, we may achieve the MSE distortion predicted by the test channel from Figure~\ref{fig:SigDelTest} (characterized in Proposition~\ref{prop:SigDelRD}) with a scalar quantizer whose rate is $I(U^{\sd}_n;U^{\sd}_n+N^{\sd}_n)+\delta(P_{ol})$, where $\delta(P_{ol})=\m{O}\left(\log\log\left(\tfrac{1}{P_{ol}}\right)\right)$.

Let $Q_{R,\sigma^2}(\cdot)$ be a uniform quantizer with quantization step $\sqrt{12\sigma^2}$ and $2^R$ quantization levels, such that the quantizer support is $[-\Gamma/2,\Gamma/2)$, where $\Gamma\triangleq 2^R\sqrt{12\sigma^2}$, see Figure~\ref{fig:scalarquant}. Our goal is to analyze the distortion and overload probability attained by a sigma-delta modulator that uses a $Q_{R,\sigma^2_{\sd}}(\cdot)$ quantizer, as a function of $R$ and $\sigma^2_{\sd}$.

Clearly, if we employ the scalar sigma-delta modulator on a long enough input sequence, an overload event will eventually occur. As discussed above, the effects of overload errors can be amplified due to the feedback loop, and in this case the average MSE may significantly grow. We therefore split the input sequence into finite blocks of length $N$, and initialize the memory of the filter $C(Z)$ with zeros before the beginning of each new block. This makes sure that the effect of an overload error in the original system is restricted to the block where it occurs.

The analysis is made much simpler by introducing a subtractive \emph{dither}\cite{ramibook}. Namely, let $\{Z_n\}$ be a sequence of i.i.d. random variables uniformly distributed over the interval $[-\sqrt{12\sigma^2_{\sd}}/2,\sqrt{12\sigma^2_{\sd}}/2)$. In order to quantize $U^{\sd}_n$, we add $Z_n$ to it before applying the quantizer, and subtract $Z_n$ afterwards, such that the obtained result is $Q_{R,\sigma^2_{\sd}}(U^{\sd}_n+Z_n)-Z_n$. Adding and subtracting $U^{\sd}_n$, we get $U^{\sd}_n+\left(Q_{R,\sigma^2_{\sd}}(U^{\sd}_n+Z_n)-(U^{\sd}_n+Z_n)\right)$, and the quantization error is therefore
\begin{align}
N_n&\triangleq Q_{R,\sigma^2_{\sd}}(U^{\sd}_n+Z_n)-(U^{\sd}_n+Z_n)\label{QuantNoise1}
\end{align}

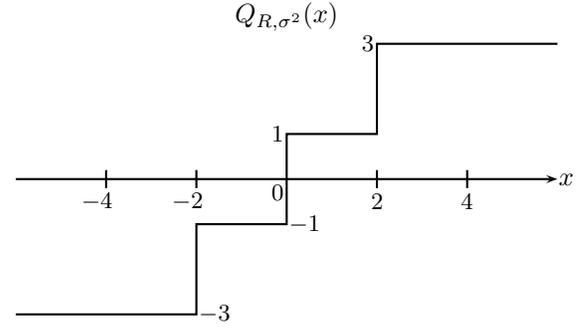
\begin{figure}[]
\begin{center}
\psset{unit=0.6mm}
\begin{pspicture}(0,-33)(125,40)

\psline{->}(0,0)(120,0)\rput(122,0){$x$}
\psline(20,2)(20,-2)\rput(18,-5){\small$-4$}
\psline(40,2)(40,-2)\rput(38,-5){\small$-2$}
\rput(58,-3){\small$0$}
\psline(80,2)(80,-2)\rput(80,-5){\small$2$}
\psline(100,2)(100,-2)\rput(100,-5){\small$4$}
\psline(0,-30)(40,-30)(40,-10)(60,-10)(60,10)(80,10)(80,30)(120,30)
\rput(44,-30){\small$-3$}
\rput(64,-10){\small$-1$}
\rput(58,10){\small$1$}
\rput(78,30){\small$3$}
\rput(60,36){$Q_{R,\sigma^2}(x)$}
\end{pspicture}
\end{center}
\caption{An illustration of $Q_{R,\sigma^2}(\cdot)$ for $R=2$ and $\sigma^2=1/3$.}
\label{fig:scalarquant}
\end{figure} 

\begin{figure*}[]
\begin{center}
\psset{unit=0.6mm}
\begin{pspicture}(0,0)(270,55)
\rput(0,30){
\rput(0,1){$X^{\sd}_n$}\psline{->}(5,0)(20,0)\pscircle(25,0){5}\rput(25,0){$\Sigma$}\rput(18,-4){$-$}\psline{->}(30,0)(95,0)\rput(85,4){$U^{\sd}_n$}
\pscircle(100,0){5}\rput(100,0){$\Sigma$}\rput(95,7){$+$}
\psline{->}(105,0)(115,0)\psframe(115,-5)(140,5)\rput(127,0){\scriptsize$Q_{R,\sigma^2_{\sd}}(\cdot)$}\psline{->}(140,0)(150,0)
\pscircle(155,0){5}\rput(155,0){$\Sigma$}\rput(150,7){$-$}

\rput(100,23){$Z_n$}\psline{->}(100,20)(100,5) \psline{->}(100,12)(155,12)(155,5)
\psline{->}(160,0)(195,0)\rput(178,4){$U^{\sd}_n+N^{\sd}_n$}
\psframe(195,-13)(245,13)\rput(220,16){$H(\omega)$}
\psline{->}(188,0)(188,-20)(90,-20)\pscircle(85,-20){5}\rput(85,-20){$\Sigma$}\rput(78,-16){$-$}\psline{->}(85,0)(85,-15)
\psline{->}(80,-20)(55,-20)\rput(73,-24){$N^{\sd}_n$}
\psframe(35,-25)(55,-15)\rput(45,-20){$C(Z)$}\psline{->}(35,-20)(25,-20)(25,-5)
\rput(195,2.5){
\psline{->}(5,-5)(45,-5)\rput(45,-8){$\omega$}
\psline(18,-5)(18,5)(32,5)(32,-5)\rput(34,5){\small $1$}
\rput(16,-10){$-\frac{\pi}{L}$}\rput(32,-10){$\frac{\pi}{L}$}
}
\psline{->}(245,0)(255,0)\rput(261,1){$\hat{X}^{\sd}_n$}
}
\end{pspicture}
\end{center}
\caption{A sigma-delta modulator with a dithered scalar uniform quantizer. The input is assumed to be over-sampled at $L$ times the Nyquist rate, and the dither sequence $\{Z_n\}$ is assumed to be an i.i.d. sequence of random variables uniformly distributed over the interval $\left[-\sqrt{12\sigma^2_{\sd}}/2,\sqrt{12\sigma^2_{\sd}}/2\right)$ and statistically independent $\{X^{\sd}_n\}$.}
\label{fig:ScalarSigDel}
\end{figure*}
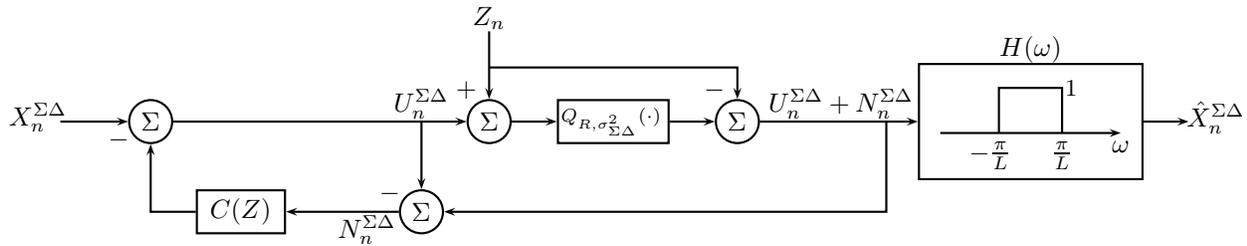 

The main result in this section is the following.
\begin{theorem}\label{thm:scalarSigDel}
Let $D$ be the MSE distortion attained by the test channel in Figure~\ref{fig:SigDelTest} with a filter $C(Z)$ of finite length, and $I(U^{\sd}_n;U^{\sd}_n+N^{\sd}_n)$ the scalar mutual information between the input and output of the AWGN channel in the same figure. For any $0<P_{ol}<1$ the scalar sigma-delta modulator from Figure~\ref{fig:ScalarSigDel} applied on a sequence of $N$ consecutive source samples with quantization rate $R=I(U^{\sd}_n;U^{\sd}_n+N^{\sd}_n)+\delta(P_{ol})$ attains MSE distortion smaller than
\begin{align}
\frac{D(1+o_N(1))}{1-P_{ol}},\nonumber
\end{align}
given that overload did not occur. In addition, the overload probability is smaller than $P_{ol}$, where $o_N(1)\rightarrow 0$ as $N$ increases, and
\begin{align}
\delta(P_{ol})\triangleq\frac{1}{2}\log\left(-\frac{2}{3}\ln\frac{P_{ol}}{2N} \right).\label{deltaPe}
\end{align}
\end{theorem}

\begin{proof}
Let $\tilde{Q}_{\sqrt{12\sigma^2}\ZZ}(x)$ be the operation of rounding $x$ to the nearest point in the (infinite) lattice $\sqrt{12\sigma^2}\ZZ$. It is easy to verify that for any $x\in[-\Gamma/2,\Gamma/2)$ we have
\begin{align}
Q_{R,\sigma^2}(x)=\tilde{Q}_{\sqrt{12\sigma^2}\ZZ}\left(x+\dfrac{\sqrt{12\sigma^2}}{2}\right)-\dfrac{\sqrt{12\sigma^2}}{2}\label{quantlatteq}.
\end{align}
Applying~\eqref{QuantNoise1} therefore yields that if overload did not occur in the $n$th sample, i.e., if $|U^{\sd}_n+Z_n|\leq \Gamma/2$, we have
\begin{align}
N_n&=\tilde{Q}_{\sqrt{12\sigma^2_{\sd}}\ZZ}\left(U^{\sd}_n+Z_n+\sqrt{3\sigma^2_{\sd}}\right)\nonumber\\
&-\left(U^{\sd}_n+Z_n+\sqrt{3\sigma^2_{\sd}}\right).\label{errornool}
\end{align}

Dealing with the overload event of the quantizer directly is rather involved. Instead, as done in~\cite{bys14}, we first consider a \emph{reference system} with an infinite-support quantizer ($R=\infty$) and analyze its performance. If the magnitude of the input to the infinite-support quantizer never exceeds $\Gamma/2$ within the processed block, then clearly the reference system is completely equivalent to the original system within this block. Thus, it suffices to find the average distortion of the reference system and the probability that the input to its quantizer exceeds $\Gamma/2$ within a block. In what follows we will therefore assume that the quantization noise is given by~\eqref{errornool} regardless of whether or not $|U^{\sd}_n+Z_n|\leq \Gamma/2$, and account for the overload probability later.

Assuming that the dither sequence $\{Z_n\}$ is drawn statistically independent of the process $\{X^{\sd}_n\}$, the Crypto Lemma, see, e.g.~\cite[Lemma 4.1.1]{ramibook}, implies that $\{N_n\}$ is an i.i.d. sequence of random variables uniformly distributed over the interval $[-\sqrt{12\sigma^2_{\sd}}/2,\sqrt{12\sigma^2_{\sd}}/2)$, statistically independent of $\{X^{\sd}_n\}$. Note that $N_n$ has zero mean and variance $\sigma^2_{\sd}$.
Following this reasoning, the reference sigma-delta data converter depicted in Figure~\ref{fig:ScalarSigDel} (with an infinite-support quantizer) is equivalent to the test channel from Figure~\ref{fig:SigDelTest} with $N^{\sd}_n\sim\Unif\left([-\sqrt{12\sigma^2_{\sd}}/2,\sqrt{12\sigma^2_{\sd}}/2)\right)$ instead of $N^{\sd}_n\sim\mathcal{N}(0,\sigma^2_{\sd})$. Thus, the average MSE distortion attained by the reference scalar sigma-delta modulator from Figure~\ref{fig:ScalarSigDel} is as given in Proposition~\ref{prop:SigDelRD} up to a multiplicative factor of $1+o_N(1)$ that accounts for edge effects. These effects are the by-product of the operation of nulling the filter memory at the beginning of each new block, which incurs temporal non-stationarities. In particular, if the filter $C(Z)$ has $L$ taps, then only after $L$ samples within the block the statistics of the process $\{U^{\sd}_n\}$ will converge to its stationary distribution. However, if the block length is sufficiently large w.r.t. the filter length and the inverse of the MSE distortion, the influence  of these effects vanishes. %We are left with the task of dealing with the probability that the input to the quantizer exceeds $\Gamma/2$ in magnitude at some time instance within the processed block.

Next, we turn to analyze the probability that an overload error occurs within a block of length $N$, as a function of $R$ and $I(U^{\sd}_n;U^{\sd}_n+N^{\sd}_n)$. Since this event is equivalent to the event that at the reference system some input to the quantizer exceeds $\Gamma/2$ in magnitude within the block, it suffices to upper bound the probability of the latter event.

Assume the reference scalar sigma-delta modulator from Figure~\ref{fig:ScalarSigDel} is applied to a vector $\bx^{\sd}=[X^{\sd}_1,\ldots,X^{\sd}_N]$ of $N$ consecutive samples of the process $\{X_n^{\sd}\}$, where the memory of the filter $C(Z)$ is initialized with zeros. Define the event $\text{OL}_k\triangleq \{|U^{\sd}_k+N^{\sd}_k|>\Gamma/2\}$ and the event $\text{OL}\triangleq\cup_{k}^N \text{OL}_k$. By the union bound, we have
\begin{align}
P_{ol}\triangleq \Pr(\text{OL})\leq\sum_{k=1}^N \Pr\left(\text{OL}_k\right).\label{oltoterror}
\end{align}

The random variable $U^{\sd}_k+N^{\sd}_k=X^{\sd}_k+(\delta_k-c_k)*N_k^{\sd}$ is a linear combination of a Gaussian random variable $X_k^{\sd}$ and statistically independent uniform random variables $\{N_k^{\sd}\}$. In~\cite[Lemma 4]{oe15it} the probability that a random variable of this type exceeds a certain threshold was bounded in terms of its variance. Applying this bound to $U^{\sd}_k+N^{\sd}_k$ yields
\begin{align}
\Pr\bigg(|U^{\sd}_k&+N^{\sd}_k|>\Gamma/2\bigg)\leq 2 \exp\left\{-\dfrac{\Gamma^2}{8\mathbb{E}(U^{\sd}_k+N^{\sd}_k)^2} \right\}\nonumber\\
&=2 \exp\left\{-\dfrac{12\sigma^2_{\sd}2^{2R}}{8\left(\mathbb{E}(U^{\sd}_k)^2+\mathbb{E}(N^{\sd}_k)^2\right)} \right\},\nonumber
\end{align}
where in the last equality we have used the definition of $\Gamma$ and the fact that $U_k^{\sd}$ and $N_k^{\sd}$ are statistically independent.
Equivalently, we may write
\begin{align}
\Pr\bigg(\text{OL}_k \bigg)&\leq 2 \exp\left\{-\dfrac{12\sigma^2_{\sd}2^{2R}}{8\sigma^2_{\sd}\left(1+\frac{\mathbb{E}(U^{\sd}_k)^2}{\sigma^2_{\sd}}\right)} \right\}\nonumber\\
&=2 \exp\left\{-\dfrac{3}{2}2^{2\left(R-\dfrac{1}{2}\log\left(1+\frac{\mathbb{E}(U^{\sd}_k)^2}{\sigma^2_{\sd}}\right) \right)} \right\}\nonumber\\
&=2 \exp\left\{-\dfrac{3}{2}2^{2\left(R-I\left(U^{\sd}_k;U^{\sd}_k+N^{\sd}_k \right)\right)} \right\},\label{olerrork}
\end{align}
where we have used~\eqref{Iineqsigdel} in the last equality. Substituting~\eqref{olerrork} into~\eqref{oltoterror} gives
\begin{align}
P_{ol}\leq 2\sum_{k=1}^N \exp\left\{-\dfrac{3}{2}2^{2\left(R-I\left(U^{\sd}_k;U^{\sd}_k+N^{\sd}_k \right)\right)} \right\}.\label{oltoterrorMI}
\end{align}
Note that $\mathbb{E}(U^{\sd}_k)^2=\sigma^2_{X}+\sigma^2_{\sd}\sum_{m=1}^k c_k^2$ is monotonically nondecreasing in $k$ and is given by~\eqref{UnSigDelVar} for values of $k$ that are greater than the length of the filter $c_k$. We can therefore further bound~\eqref{oltoterrorMI} as
\begin{align}
P_{ol}\leq 2N \exp\left\{-\dfrac{3}{2}2^{2\left(R-I\left(U^{\sd}_n;U^{\sd}_n+N^{\sd}_n \right)\right)} \right\},\label{oltoterrorMI1}
\end{align}
where $I\left(U^{\sd}_n;U^{\sd}_n+N^{\sd}_n \right)$ is as given in Proposition~\ref{prop:SigDelRD}.
To summarize, we have shown that the reference system achieves the same MSE distortion $D$ as characterized by Proposition~\ref{prop:SigDelRD} up to a $1+o_N(1)$ multiplicative term, and that the probability that one of the quantizer input samples exceeds $\Gamma/2$ in magnitude within a block of length $N$, is bounded by~\eqref{oltoterrorMI1}. For our original system whose quantizer has finite support of $[-\Gamma/2,\Gamma/2)$, this means that the overload probability is also upper bounded by the RHS of~\eqref{oltoterrorMI1}. Moreover, the average distortion it achieves if overload did not occur is the same as that of the reference system conditioned on the event that $\text{OL}$ did not occur. Denote this conditioned expected distortion by $D_{\overline{\text{OL}}}$ and the expected distortion conditioned on the event that $\text{OL}$ did occur by $D_{\text{OL}}$. For the reference system, we have
\begin{align}
D(1+o_N(1))=\Pr(\overline{\text{OL}})D_{\overline{\text{OL}}}+\Pr(\text{OL})D_{\text{OL}}\geq \Pr(\overline{\text{OL}})D_{\overline{\text{OL}}},\nonumber
\end{align}
and therefore
\begin{align}
D_{\overline{\text{OL}}}\leq \frac{D(1+o_N(1))}{1-P_{ol}}.\nonumber
\end{align}

This shows that the scalar sigma-delta system from Figure~\ref{fig:ScalarSigDel}, whose quantizer has limited support $[-\Gamma/2,\Gamma/2)$, with $R=I(U^{\sd}_n;U^{\sd}_n+N^{\sd}_n)+\delta(P_{ol})$ achieves the same average MSE distortion as the test channel from Figure~\ref{fig:SigDelTest} up to a multiplicative factor of $(1+o_N(1))/(1-P_{ol})$, with block error probability smaller than $2N\exp\left\{-\dfrac{3}{2}2^{2\delta}\right\}$. Thus, Proposition~\ref{prop:SigDelRD} characterizes the rate-distortion tradeoff achieved by the scalar sigma-delta system up to the aforementioned factor and a constant rate penalty $\delta(P_{ol})$, that depends on the target overload error probability. To be more precise, for any $0<P_{ol}<1$, taking the rate penalty as in~\eqref{deltaPe} guarantees that the overload error probability is smaller than $P_{ol}$.
\end{proof}

\section*{Acknowledgements}
We thank Jan {\O}stergaard and Ram Zamir for their valuable comments on an earlier version of this manuscript.

\bibliographystyle{IEEEtran}
\bibliography{SigDel_bib}

\end{document}